\documentclass[11pt]{article}
\usepackage[affil-it]{authblk}
\usepackage{graphicx}
\usepackage{hyperref}
\usepackage{natbib}
\usepackage[utf8]{inputenc}
\usepackage{latexsym,amsmath,amssymb, amsthm}

\newtheorem{definition}{Definition}[section]
\newtheorem{theorem}[definition]{Theorem}
\newtheorem{lemma}[definition]{Lemma}
\newtheorem{corollary}[definition]{Corollary}

\newcommand{\defeq}{\mathrel{\mathop:}=}

\begin{document}

\title{Analysis of the Sybil defense of Duniter-based cryptocurrencies
}

\author{Lucas Isenmann}

\affil{LIRMM, Université de Montpellier, 161 rue Ada, 34095, Montpellier, France}

\date{}

\maketitle

\begin{abstract}
Duniter-based cryptocurrencies, which are providing a kind of universal basic income, are using a system called "Web of Trust" based on a social network whose evolution is subject to graph theoretical rules, time constraints and a licence in order to avoid large Sybil attacks.
We investigate in this article the largest size of a Sybil attack that a simplified version of the graph theoretical rules of a Web of Trust can undergo depending on the number of attackers and on the parameters of the system.
We show that even if in theory, without considering social and time constraints, this system cannot in general prevent huge attacks, in the real-world case of a Duniter-based cryptocurrency (with thousands of users), the system can prevent attacks of large size with only graph theoretical rules.


\end{abstract}

\section{Introduction}

    Cryptocurrencies providing different kind of universal basic income have appeared recently: Duniter \cite{Duniter}, Good Dollar~\cite{GoodDollar}, Circles~\cite{Circles}, PopCoin~\cite{borge2017proof} and Encointer~\cite{brenzikofer2019encointer}.
    All these cryptocurrencies face the same problem where badly intentioned users may create fake users in order to receive several sources of money (in addition of their own source of money).
    In the case of reputation system such attacks are called Sybil attacks \cite{douceur2002sybil}.
    As it deals with a crucial power (the creation of money), it is of very high importance that each account created on such cryptocurrencies corresponds to exactly one physical person.
    This problem is also faced by e-voting systems.
    To solve this problem of identification, the Proof of Personhood has been for example developped for PopCoin~\cite{borge2017proof} and Encointer~\cite{brenzikofer2019encointer} where users enter the system by taking part at a physical gathering where they each receive an identifier (a proof-of-personhood token).
    
    Duniter is a general blockchain for managing cryptocurrencies providing a kind of universal basic income.
    Each cryptocurrency created through Duniter has its own Sybil defense, with its own parameters, called "Web of Trust" in order to limit the creation of fake accounts.
    This system is only based on graph theoretical rules, time constraints and on a licence.
    The underlying directed graph of the Web of Trust is such that vertices correspond to users and arcs correspond to what is called "certifications".
    If there is an arc from $A$ to $B$, we say that $A$ has "certified" $B$.
    According to the licence, a user can certify another user only if they know each other "enough" (as defined by the licence).
    Therefore this system can be seen as a social network where there occurs a propagation of trust.
    The first Duniter-based cryptocurrency is called Ğ1 and its Web of Trust was already studied as a social network \cite{gensollen2020you}.
 
    We will now focus our study only on the graph theoretical part of this system in order to determine in what order the system is able to prevent Sybils attack with only its graph theoretical rules.
    This system therefore generates a complex network representing the propagation of trust among users of the cryptocurrency and whose evolution is restricted to a fixed set of rules.

\subsection{The Web of Trust of a Duniter-based cryptocurrency}

For a comprehensive description of the system see \cite{Duniter} which is a protocol consisting of around one hundred of rules mixing technical rules of the blockchain and of the Web of Trust.
Because of the complexity of the protocol, we summarize the rules of this system as follows.
The Web of Trust of a Duniter-based cryptocurrency is a directed graph $G$ where the vertices $V(G)$ correspond to users and the arcs $A(G)$ correspond to certifications that users give to other users.
The system has the following parameters: three positive integers $\delta, \Delta, d$ and a real number $\alpha \in [0,1]$.
In the remaining of the article, these parameters are fixed.

\begin{definition}
    A vertex $v$ is said to be \emph{referent} if $d^+(v) \geq n^{-1/d} $ and $d^-(v) \geq n^{-1/d}$.
    We denote by $r(G)$ the total number of referent vertices in $G$.
    For any integer $i$ and any set $X$ of vertices, we denote by $r_i(X)$  the number of referents vertices at distance at most $i$ to $X$.
    Formally, $r_i(X) = | \{ v \in V(G) \colon d(v,X) \leq i \} |$.
\end{definition}

Referents vertices will play the role of "central" vertices to which other vertices should be "not too far".
At any time, the graph should satisfy the following rules:

\begin{definition}
    The digraph $G$ is said to be \emph{valid} if it satisfies the following rules, for every vertex $v$ of $V(G)$:
    \begin{itemize}
        \item (in-degree rule) $d^-(v) \geq \delta$,
        \item (out-degree rule) $d^+(v) \leq \Delta$,
        \item (distance rule) $r_{d}(v) \geq \alpha r(G)$.
    \end{itemize}

\end{definition}

The Web of Trust is initialized with any valid digraph.
After the initialization, only two operations can be triggered by users: \emph{vertex creations} and \emph{arc creations}.
Only certain group of users can trigger a vertex creation.
Such groups are called \emph{introducers}:


\begin{definition}[Introducer group]
    Let $X$ be a subset of vertices.
    We say that $X$ is an \emph{introducer group} if it satisfies the following rules:
\begin{itemize}
    \item (in-degree rule) $|X| \geq \delta$
    \item (out-degree rule) for any vertex $v$ of $X$, $d^+(v) < \Delta$
    \item (distance rule) $r_{d-1}(X) \geq \alpha r(G)$.
\end{itemize}
\end{definition}

At any moment, an introducer group of users can effectively create a vertex according to the following operation.
This operation represent the fact that a new user joins the Web of Trust and that the previous group is certifying this new user.

\begin{definition}[Vertex creation]
Given an introducer group $X$ of vertices, the operation on $G$ consisting in adding a vertex $v$ and arcs from every vertex of $X$ to $v$ is called a \emph{vertex creation}.
\end{definition}

Remark that after a vertex creation from an introducer group $X$ of vertices, the in-degree and out-degree rules are satisfied by every vertices of the new graph.
Nevertheless, the referent status of old vertices may change: as the number of vertices has increased by $1$, the condition to be referent is harder to be satisfied, thus some vertices can loose their referent status.
Furthermore, vertices among the introducer group $X$ may become referents as they have given one more certification.
Therefore the distance rule is not necessarily satisfied by the new graph.

The second operation is described as follows.
At any moment, a user $v$, which is of out-degree at most $\Delta-1$ can add an arc, or \emph{give a certification}, from $v$ to any other user (which already exists in the Web of Trust).

\begin{definition}[Arc creation]
    Given a vertex $v$ such that $d^+(v) < \Delta$ and any vertex $x$, the operation consisting in adding an arc from $v$ to $x$ is called an \emph{arc creation}.
\end{definition}

Remark that after an arc creation, every vertex different from the two vertices of the arc does not change its referent status (as the number of vertices has not changed).
But, the two vertices of the arc can become referent if they were not.
Therefore, even if the in-degree and out-degree rules are still satisfied in the new graph, it is not necessarily the case for the distance rule.
Thus the new graph is not necessarily valid.

\paragraph{Differences between the protocol and this theoretical study.}

In the Web of Trust protocol, if after a vertex creation or an arc creation, the new graph is not valid because a vertex does not satisfy anymore the distance rule, then this vertex is automatically removed from the Web of Trust and recursively until the graph is valid.
At the difference with the Web of Trust protocol, the previous elimination process is not taken in account because the theoretical study of Sybil attacks would be too complicated.

Furthermore, the Web of Trust protocol is using time constraints in order to slow down the evolution of possible Sybil attacks.
For example, it is only possible to trigger a vertex creation or an arc creation every five days.
In this study, we do not take these time constraints in account and therefore we suppose that the operations occur immediately.

\section{Maximum size of a Sybil attack}

\subsection{Sybil attacks}

A Sybil attack is the creation of fake identities controlled by some members which will be called attackers.
The goal is to limit the size of the attack and to slower it.
In the rest of the article, $A$ will denote a set of vertices of size at least $\delta$.

\begin{definition}[Sybil attack]
A \emph{Sybil attack} from $A$ is a sequence of vertex creations and arc creations using vertices from $A$ and from the created vertices.
The vertices of $A$ are called the \emph{attackers} and the vertices created during the sequence are called the \emph{fake vertices}.
The number of fake vertices is called the \emph{size} of the Sybil attack.
\end{definition}

The natural question that arises is: how big can be a Sybil attack?
In order to answer this question, we will use the following notions.
Given a Sybil attack from $A$, we denote by $F_i$ the set of fake vertices at distance $i$ from $A$ and $F_0 \defeq A$.
Furthermore we consider the following quantity which is linked to the out-degree rule:
\begin{definition}
For any subset $X$ of vertices of $G$, we denote by $c(X)$ the total number of available certifications of $X$.
Formally it is defined as follows: $c(X) \defeq |X| \cdot \Delta - \sum_{x \in X} d^+(x)$.
\end{definition}

Finally, remark the analogy with the degree-diameter problem which has been solved by the Moore theorem.
This theorem asserts that a digraph whose maximum distance between two vertices (called the diameter) and whose maximum out-degree are bounded cannot be arbitrary large.

\begin{theorem}[Moore \cite{Moore}]
Let G be a digraph such that $d$ is the diameter of the graph and such that $\Delta$ is the maximum out-degree.
Then $|G| \leq \Delta^d$.
\end{theorem}

Here, we investigate the case where the out-degree is bounded by a constant and the number of "far in-neighbors" (more exactly referents vertices which are at distance at least $d$ to a vertex) is bounded by a linear function of the size of the graph.
As in the case of the Moore theorem, we show that these conditions bound the size of the attack.

\subsection{Upper bounds}


We now prove a general upper bound on the number of fake vertices at distance at most $d$ that can be created in a Sybil attack from $A$.

\begin{theorem}
\label{theorem:upper_bound}
For any $k \geq 1$, there are at most $g_k \defeq c(A) \frac{(\Delta-\delta+1)^k-1}{\Delta-\delta}$ fake vertices at distance at most $k$ from $A$.
\end{theorem}
\begin{proof}
We define $\alpha \defeq \Delta - \delta +1$.
We can see $c(F_i)$ and $|F_i|$ as variables that evolve during the Sybil attack.
We now define the following linear function on these variables:
\[
f \defeq \sum_{i=0}^d \left( c(F_i) \sum_{j=0}^{d-i} \alpha^j \right) + \sum_{i=1}^{d+1} |F_i|
\]
Note that the weights of the variables $|F_i|$ are $1$ and that those of the variables $c(F_i)$ are decreasing.
Let us show that in moment of the attack, $f \leq  c(A) \sum_{j=0}^d \alpha^j$.
At the beginning of the attack, we have $c(F_0) = c(A)$ and $|F_i| = c(F_i) = 0$ for every $i$.
Thus, at the beginning, we have $f = c(A) \sum_{i=0}^{d} \alpha^j$ and the previous inequality is satisfied.
Let us now check that the quantity $f$ decreases for any Sybil attack.

The creation of a vertex of $F_{i+1}$ must use at least one certification from a vertex of $F_i$ and possibly from vertices of $F_k$ for $k \geq i+1$ but not from vertices of $F_k$ for $k < i$.
Thus by creating a vertex of $F_{i+1}$ we note $p_k$ the number of certifications used from vertices of $F_k$ (satisfying therefore the equality $\sum_{k} p_k \geq \delta$).
We first treat the case where where $p_k = 0$ for every $k > i+1$.
As the creation of a vertex of $F_{i+1}$ increases the number of available certifications of $F_{i+1}$ by $\Delta$ we get the following table which sums up the difference on the variables:
\[
\begin{array}{cccc}
    c(F_i) & |F_{i+1}|  & c(F_{i+1}) & |F_{i+2}|  \\
    -p_i & +1  & - (\delta-p_i) + \Delta & +0
\end{array}
\]
In this case we define $p \defeq p_i$.
Let us check that the quantity $f$ decreases for such a creation.
As $\alpha = \Delta -\delta +1$, the difference for $f$ is 

\begin{align*}
   & -p \left( \sum_{j=0}^{d-i} \alpha^j \right) + 1 + (\Delta - \delta +p) \sum_{j=0}^{d-i-1} \alpha^j \\
    =& -(p-1) \left( \sum_{j=0}^{d-i} \alpha^j \right) - \left( \sum_{j=0}^{d-i} \alpha^j \right) + 1 \\
    & + (\Delta - \delta +1) \left(\sum_{j=0}^{d-i-1} \alpha^j \right) + (p-1) \left(\sum_{j=0}^{d-i-1} \alpha^j \right) \\
    = &  -(p-1) \left( \sum_{j=0}^{d-i} \alpha^j \right) + (p-1) \left(\sum_{j=0}^{d-i-1} \alpha^j \right) \leq 0
\end{align*}

The general case can be checked in the same manner as the weights are decreasing with $c(F_i)$.

We conclude that $f \leq c(A) \sum_{i=0}^d \alpha^j$ during the Sybil attack.
We deduce the announced inequality as $\alpha = \Delta -\delta +1$ and as the weights and the variables are non-negative.
\end{proof}

We now show that is bound is kind of tight if $\delta < \Delta$ by describing an explicit Sybil attack strategy.
The idea is to use as less as possible the attackers.
We start by describing how this attack creates the vertices of $F_1$.

\begin{lemma}
If $c(A) \geq \binom{\delta}{2}$, then we can create $c(A) - \binom{\delta}{2}$ fake vertices at distance $1$ from $A$ such that these fake vertices have  $(c(A) - \binom{\delta}{2}) \cdot (\Delta-(\delta-1)) + \binom{\delta}{2}$ available certifications in total.
\end{lemma}
\begin{proof}
The construction consists in two steps.
The first step creates $\delta-1$ fake vertices by minimizing the use of certifications from $A$ as follows:
the first fake vertex will get certifications from $\delta$ attackers  (the choice of the $\delta$ attackers is arbitrary).
The second one will get certifications from $\delta-1$ of the attackers and from the first fake vertex, \textit{etc...}

\begin{figure}[ht]
    \centering
    \includegraphics{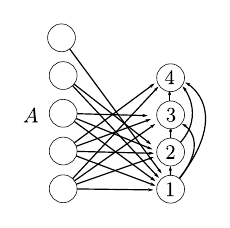}
    \caption{Step 1 of the Sybil optimal strategy (in the case where $\delta = 5$): creation of the first $\delta-1$ fake vertices. The order of creation of the fake vertices is noted on the vertices.}
\end{figure}

The second step consists in creating fake vertices with certifications from the most recent fake vertices and from one attacker (that can be chosen arbitrarily).

\begin{figure}[ht]
    \centering
    \includegraphics{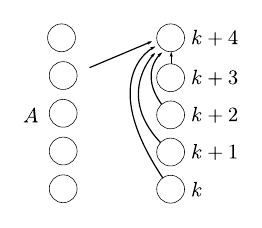}
    \caption{Step 2 of the Sybil optimal strategy: creation of a fake vertex with certifications from the $\delta-1$ most recent fake vertices and from one attacker. The order of creation of the fake vertices is noted on the right of the vertices.}
\end{figure}

Let us count how much fake vertices we can create using this strategy.
After the first step, attackers have used $(\delta-1) + (\delta-1) +  \cdots  + 1 = \delta-1 + \binom{\delta}{2}$ certifications over the $c(A)$ available.
Therefore there remains $c(A) - (\delta-1) - \binom{\delta}{2}$ certifications to give from attackers in total.
We can repeat the creation of step 2 as long as attackers have at least one certification available.
Thus we can create $c(A) - (\delta-1) - \binom{\delta}{2}$ fake vertices during step 2.
Plus the $\delta-1$ fake vertices created during step 1, we get a total of $ \delta-1 + c(A) - (\delta-1) - \binom{\delta}{2} = c(A) - \binom{\delta}{2}$ fake vertices at distance at most $1$ from $A$.

Furthermore, let us count the total number of available certifications for the fake vertices.
At the exception of the $\delta-1$ latest fake vertices  ones have given certifications to $\delta-1$ of the next fake vertices.
The $\delta-1$ latest fake vertices have given less certifications: the last one has given no certifications, the second latest one has given one certification, \textit{etc} ... 
Thus in total, there are $(c(A) - \binom{\delta}{2}) \cdot (\Delta-(\delta-1)) + \binom{\delta}{2}$ certifications available from these fake vertices.
\end{proof}

By repeating recursively $d$ times the previous Lemma, we can prove the following theorem:
\begin{theorem} \label{theorem:upper_bound_worst_case}
If $r_1(A) = r(G)$ and if $c(A) \geq \binom{\delta}{2}$, then there exists a Sybil attack of size $(c(A) - \binom{\delta}{2}) \frac{(\Delta-\delta+1)^{d}-1}{\Delta - \delta}$.
\end{theorem}

Therefore the upper bound of Theorem~\ref{theorem:upper_bound} is asymptotically optimal up to a multiplicative coefficient.
The difference between the upper bound and the latter strategy size in the multiplicative coefficient of $\binom{\delta}{2}$ is due to the fact that the $\delta$ first vertices of $F_i$ which are created need at least $\binom{\delta}{2}$ certifications from $F_{i-1}$.
This fact is not taken in account in the proof of the upper bound.
While the gain is not relevant, we conjecture that it is possible to improve the upper bound.

The upper bound of Theorem~\ref{theorem:upper_bound} can be, nearly,  achieved but only in a very extreme case: all the referents vertices must have certified all the attackers.
Therefore we prove the following theorem which gives an upper bound taking in account that referents are at different distance to the attackers.

\begin{theorem}
\label{theorem:upper_bound_with_referents}
For any $k \geq 0$, if $(r_{d-k}(A) + |A| + g_k) / (r(G) + |A| + g_k) < \alpha$, then the size of any Sybil attack from $A$ is at most $g_k$.
\end{theorem}
\begin{proof}
Suppose that there exists a fake vertex $v$ at distance $k+1$ from $A$.
We note $rf_k$ the number of fake referent vertices at distance at most $k$ from $A$.
According to Theorem~\ref{theorem:upper_bound}, $rf_k \leq |A| + g_k$ for any $i \geq 0$.
As $v$ should satisfy the distance rule, then $( r_{d-k}(A) + rf_k  ) / (r(G) +  rf_k) \geq \alpha$.
Remark that the function $x \longmapsto x/(x+r(G))$ is increasing on $[0,+\infty[$.
Thus $( r_{d-k}(A) + |A| + g_k) / (r(G) + |A| + g_k) \geq \alpha$.
This contradicts the hypothesis of the theorem.
We conclude that the Sybil attack is of size at most $g_{k}$.
\end{proof}

Remark that if $r_{d-1}(A)/r(G) < \alpha$ then the size of any Sybil attack from $A$ is $0$ meaning that $A$ is not an introducer group.
This constraint applies also for honest vertices who want to create new vertex.
Therefore this distance rule slow down the social development of a Duniter Web of Trust.

By taking $k = d$ in Theorem~\ref{theorem:upper_bound_with_referents}, we obtain the following corollary.
We can see that the size of any Sybil attack is bounded in the case where the number of referent vertices is enough large.
\begin{corollary}
If $|R| > \frac{1}{\alpha} (|A| + g_d)$, then the size of any Sybil attack from $A$ is at most $g_d$.
\end{corollary}

Nevertheless this upper bound is rather huge.
This is due to the fact that this theorem is valid in particular in the worst case where all referent vertices have certified all the attackers.

\section{Conclusion}

The first real-world Duniter-based cryptocurrency is called Ğ1 and is using the following parameters: $d = 5$, $\delta = 5$, $\Delta = 100$ and $\alpha = 0.8$.
In april 2020, there were 2590 vertices in the associated Web of Trust.
With these parameters, Theorem~\ref{theorem:upper_bound_worst_case} gives the worst case upper bound of $\sim 2 \cdot 10^{10}$ fake vertices which is rather huge compared to the size of the graph.
Nevertheless this upper bound has been obtained in the worst case where all the referent vertices have certified all the attackers which is practice not the case.
According to Dunbar's study, one person is able to maintain stable relationships with only around 150 people \cite{dunbar1992neocortex}, thus making it impossible for a group of attackers to have been certified by all the referent vertices if certifications correspond to stable relationship (which should be the case according to the licence of the Web of Trust).

In order to study the maximum size of a Sybil attack in a real Web of Trust of a Duniter currency, we consider another Sybil strategy that we called "basic".
The theoretical advantage of this one is that the number of fake vertices that can be created can be efficiently computed in function of the numbers $r_i(A)$.

Suppose for convenience that $|A| = \delta$.
The basic strategy goes as follows.
Start by creating as most as possible fake vertices with all the attackers.
Then from every group of $\delta$ fake vertices, create $\Delta$ fake vertices.
Repeat until it is not possible (as the distance rule may not be satisfied).
See Figure~\ref{fig:sybil_basic} for a scheme of a basic Sybil attack.

\begin{figure}[ht]
    \centering
    \includegraphics[width=0.4\textwidth,page=1]{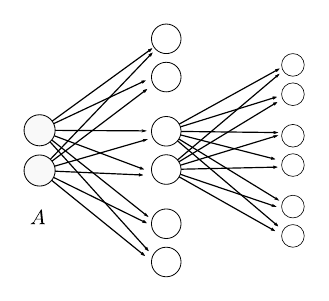}
    \caption{Example of a Sybil basic attack with $\delta = 2$, $\Delta = 6$ and $|A| = 2$.}
    \label{fig:sybil_basic}
\end{figure}

This strategy is not efficient at all, as fake vertices will struggle to satisfy the distance rule.
Indeed, as the structure is a directed tree, the only near referents for fake vertices will be their predecessors which is too much compared to the total number of referent fake vertices.

Nevertheless, we think that this strategy shows the difficulty for a Sybil attack to grow.
Results for this strategy on the real Web of Trust of the Duniter currency Ğ1 are depicted in Figure~\ref{sybil_size_basic}.
As we can see on the figure, in general $5$ users can make a basic Sybil attack of size $\sim 100$.
As it is impossible to differentiate a group of $5$ honest users introducing new honest users or a group of $5$ attackers, it should be normal that every group of $5$ vertices could introduce a fixed number of new vertices.
Only a few groups of $5$ users can reach a basic Sybil attack of size $\sim 1500$.

\begin{figure}[ht!]
    \centering
    \includegraphics[width=0.6\textwidth,page=1]{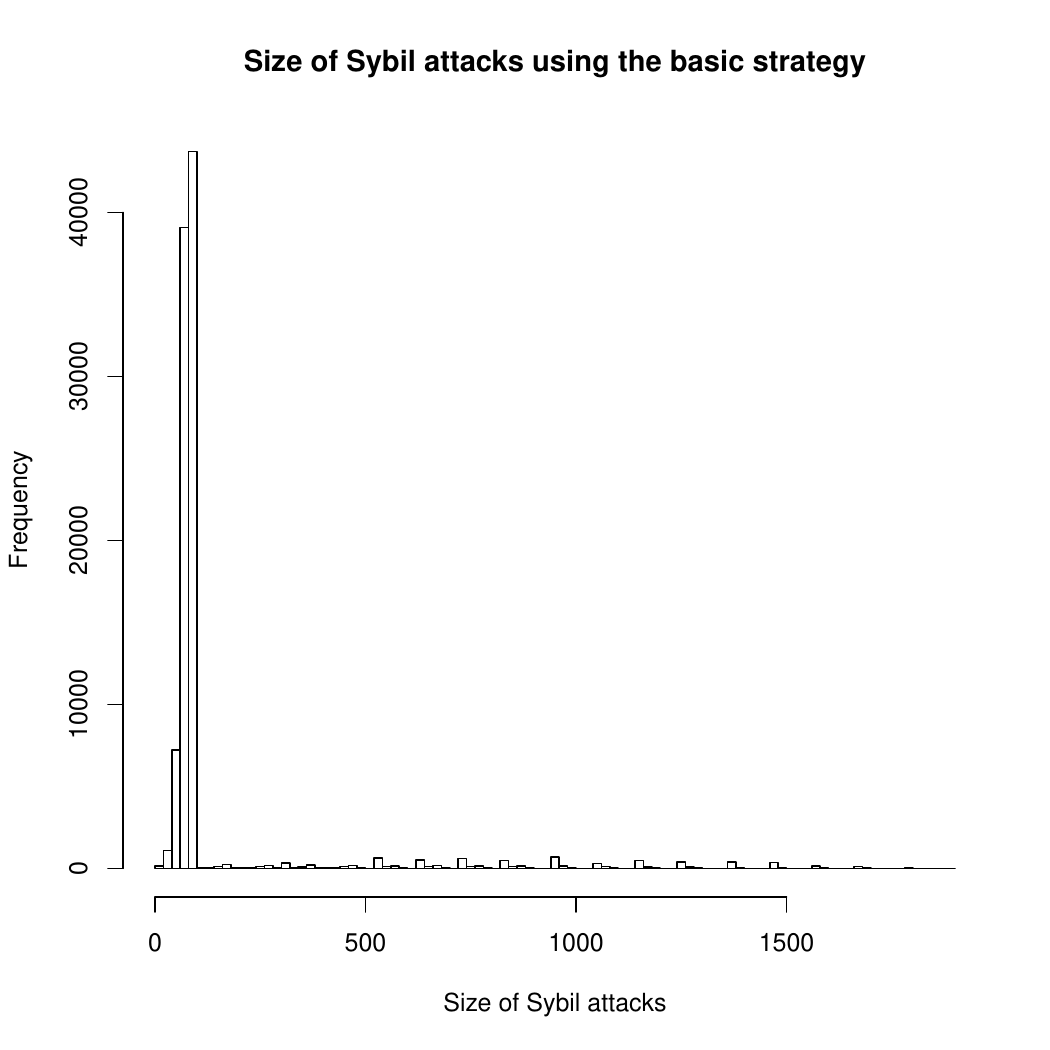}
    \caption{Statistics about the size of basic Sybil attack in the real world Web of Trust of the Duniter-based cryptocurrency Ğ1 which is of size $2590$ in april 2020 for a sample of 100.000 random sets of $5$ attackers. }
    \label{sybil_size_basic}
\end{figure}

It would be interesting to know how the maximum size of a Sybil attacks evolve with the growing of this cryptocurrency.
Furthermore, this study does not take in account that if a user has been part of an operation, then he must wait a certain time to trigger another operation.
Therefore it will slower the growth of a Sybil attack.
It would be interesting to study this speed.
Finally, we recall that  we did not take in account that vertices and arcs could disappear.
This rule should be an advantage for attackers as they may make the honest vertices disappear by creating too much referent fake vertices.

\bibliographystyle{unsrt}
\bibliography{references}

\end{document}